\documentclass[conference,letterpaper]{IEEEtran}

\usepackage[utf8]{inputenc}
\usepackage[T1]{fontenc}
\usepackage{url}
\usepackage{ifthen}
\usepackage{cite}
\usepackage[cmex10]{amsmath}
\usepackage{bm}

\usepackage{amssymb,amsfonts,amsmath,amsthm,amscd,dsfont,mathrsfs,bbold,pifont}
\usepackage{blkarray}
\usepackage{graphicx,float,psfrag,epsfig,color}

\usepackage{amssymb,amsfonts,amsmath,amsthm,amscd,dsfont,mathrsfs,bbm,tikz}
\usepackage{graphicx,float,psfrag,epsfig,color}
\usepackage{blkarray}
\usepackage{qtree}
\usepackage{comment}
\theoremstyle{plain}
\newtheorem{defin}{Definition}[section]
\theoremstyle{plain}
\newtheorem{thm}{Theorem}[section]
\theoremstyle{plain}
\newtheorem{claim}{Claim}[section]
\theoremstyle{plain}

\theoremstyle{plain}

\theoremstyle{plain}
\newtheorem{model}{Model}[section]
\theoremstyle{plain}
\newtheorem{condition}{Condition}[section]
\theoremstyle{plain}

\theoremstyle{plain}
\newtheorem{prop}{Proposition}[section]
\theoremstyle{plain}

\theoremstyle{plain}

\theoremstyle{plain}
\newtheorem{conjecture}{Conjecture}[section]
\theoremstyle{remark}

\newcommand{\distas}[1]{\mathbin{\overset{#1}{\sim}}}

\DeclareMathOperator*{\argmax}{arg\,max}
\begin{document}
\title{Polarization in Attraction-Repulsion Models}


\author{%

    \IEEEauthorblockN{Elisabetta Cornacchia (EPFL)}
    \and
    \IEEEauthorblockN{Neta Singer (Columbia University)}
    \and
    \IEEEauthorblockN{Emmanuel Abbe (EPFL)}
}

\maketitle

\begin{abstract}
This paper introduces a model for opinion dynamics, where at each time step, randomly selected agents see their opinions --- modeled as scalars in $\bm{[0,1]}$ --- evolve depending on a local interaction function. In the classical Bounded Confidence Model, agents opinions get attracted when they are close enough. The proposed model extends this 
by adding a repulsion component, which models the effect of opinions getting further pushed away when dissimilar enough. With this repulsion component added, and under a repulsion-attraction cleavage assumption, it is shown that a new stable configuration emerges beyond the classical consensus configuration, namely the polarization configuration. More specifically, it is shown that total consensus and total polarization are the only two possible limiting configurations. The paper further provides an analysis of the infinite population regime in dimension 1 and higher, with a phase transition phenomenon conjectured and backed heuristically.  


\end{abstract}


\section{Introduction}
Opinion dynamics have been widely studied in the recent years, driven in part by understanding when consensus can be reached or not \cite{Hegselmann02opiniondynamics,gomez2012bounded,Mossel_2017,cite-key,PhysRevE.97.022312}, as well as by understanding when polarization phenomenon may take place \cite{del2017modeling,Krueger_2017,doi:10.1080/0022250X.2018.1517761,deffuant2004modelling}. Many of the models that have been developed are based on binary opinions that agents update under social interactions, such as the voter model and the majority rule \cite{Gastner_2018,Choi_2019}. 
However, these models also have their limitations, such as the absence of more temperate positions. Of interest to us are models where opinions are continuous variables, such as political inclinations between Left and Right, opinions on topics, or values of a utility function in economics. The Bounded Confidence Models (\cite{Deffuant2000,Hegselmann02opiniondynamics,gomez2012bounded,PhysRevE.97.022312}) include dynamics where continuous positions are updated under binary encounters whenever the opinion dissimilarity between two participants is below a given threshold. This assumes a constructive discussion between the paired agents when they already agree closely, that causes their opinions to be further attracted to each other.

In this paper, we combine this homogenization/attraction effect with a repulsion effect. People are likely to examine opposite positions in a biased way and to repulse under disagreement or far enough opinions \cite{articleLGP}. Hence, the interaction between individuals with very different opinions may result in an even larger separation \cite{deffuant2004modelling}. 

This paper introduces a new class of models describing pairwise interaction under a common dynamic.
The encounters between pairs of individuals are governed by a random selection and a two dimensional function (the interaction function) defines the updated positions of the two agents. We consider functions with an attraction effect if the opinions dissimilarity is below a given threshold $\tau$, similarly to the Bounded Confidence Model, but we also add a repulsion effect if the opinions dissimilarity is above $\tau$. Under such pairwise interactions, are the opinions of a crowd converging to  stable configurations that can be characterized? We shall see that under some hypotheses, these are of two kinds: the total consensus configuration, as in the Bounded Confidence Model, but also the polarization configuration. We also note that while the setup is different than the polarization phenomenon in polar coding \cite{Arikan09}, some of the tools (e.g. movement at non-stable configurations, stability at extremes, martingale convergence theorem) are similarly relevant.  

\subsection{Our Model} \label{sec:ourmodel}
Consider a population of $n$ agents whose opinions lie in $I=[0,1]$. Let $ \Phi_t = ( \phi_t^1,...,\phi_t^n) \in  I^n$ denote the state of the system at time $t \in \mathbb{N}$. Here $\phi_t^i \in [0,1]$ denotes the opinion of agent $i$ at time $t$.

\begin{defin}[Interaction function]
Let $f_i: I^2 \rightarrow I$, $i=1,2$, be measurable functions. We say that $f=(f_1,f_2)$ is an interaction function if for any $x,y \in I$
\begin{align}
\label{for:orderinvariance}
f_1(x,y) = f_2(y,x).
\end{align}
\end{defin}
The interaction function determines the impact of the pairwise encounter on the opinions of the two agents, depending on the opinions they had before. It will usually be a function of the relative distance between the two opinions.
In (\ref{for:orderinvariance}) we want the interaction of two agents to be independent of their order.

The random process evolves as follows. At time $t=0$, the opinions $\phi_0^i$, $i \in [n]$, are drawn iid under some probability distribution $D_0$ defined on $I$. Let $f$ be an interaction function. At each time step $t\geq 1$, we choose a random pair of agents uniformly at random from ${[n] \choose 2}$ and independently from the other time steps, say $(i_t,j_t)$, and we make them interact with each other through $f$. Then, $\phi_{t+1}^{i_t} = f_1(\phi_t^{i_t},\phi_t^{j_t})$ and $\phi_{t+1}^{j_t} = f_2(\phi_t^{i_t},\phi_t^{j_t})$. The opinions of the other agents stay unchanged. We denote by $\left(\{\Phi_t\}_{t \in \mathbb{N}}, f, D_0 \right) $ the random interaction process associated to the interaction function $f$ and the initial distribution $D_0$. There are two sources of randomness in this process. The first one comes from the initial distribution from which the agents are sampled at time $0$; the second one emerges at each step from the random selection of the pair of interacting agents. We describe how the Bounded Confidence Model (\cite{Deffuant2000,Hegselmann02opiniondynamics,gomez2012bounded,PhysRevE.97.022312}) is captured by our framework. 


\begin{model}[Bounded Confidence Model,\cite{Deffuant2000,Hegselmann02opiniondynamics,gomez2012bounded,PhysRevE.97.022312}]
\label{ex:BC}
Assume $D_0 = \mathcal{U}([0,1])$, the uniform distribution in $[0,1]$. Let $\tau,\lambda$ be in $\left(0, 1 \right)$. Consider
\begin{align*}
f(x,y) := \left\{
                \begin{array}{ll}
                 \left( x + \frac{\lambda}{2} (y-x), y + \frac{\lambda}{2} (x-y)	\right) &\text{if } |x-y| \leq \tau,\\
                  (x,y) &\text{if } |x-y| > \tau.
                \end{array}
              \right.
\end{align*}
This model assumes that if two agents with relatively similar opinions encounter, they have a constructive discussion and their opinions end up being closer to each other. On the other hand, two agents with relatively distant opinions are unable to interact, and the encounter has no effect on their opinions.
\end{model}
We now introduce the following model.
\begin{model}[Attraction-Repulsion Model]
\label{ex:taumod}
Let $D_0 = \mathcal{U}([0,1])$. Without loss of generality, assume $x \leq y$. The remaining cases follow from (\ref{for:orderinvariance}). Let $\tau, \lambda, \mu$ be in $(0,1)$. Consider
\begin{align*}
\label{for:ftau}
f(x,y) := \left\{
                \begin{array}{ll}
                 \left( x + \frac{\lambda}{2}(y-x),  y + \frac{\lambda}{2}(x-y) 	\right) &\text{if } |x-y| \leq \tau,\\
                  \left( x - \mu x,  y + \mu (1-y) \right)  &\text{if } |x-y| > \tau.
                \end{array}
              \right.
\end{align*}
Here we assume that if two agents with similar opinions encounter, they reduce their distance by a factor $1-\lambda$, similarly to the example before. However, we now also assume that a discussion between two agents with far enough opinions causes an even larger separation of the two parties, hence the distances with the respective extremes are reduced by a factor $1-\mu$. 
\end{model}
\begin{figure}[h]
    \centering
    \includegraphics[width=0.24\textwidth]{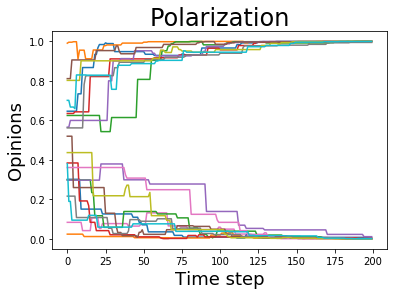}
    \includegraphics[width=0.24\textwidth]{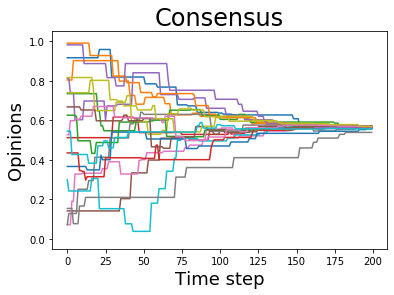}
    \caption{Dynamics of the Attraction-Repulsion model over time. For both plots we took $n=20$ agents, and $\lambda =\mu = 1/2 $. The plot on the left shows a simulation where $\tau = 0.4$. 
    The plot on the right was obtained taking $ \tau = 0.6$. 
    }
    \label{fig:simtaumod}
\end{figure}
In both examples, $\tau$ is a measure of the tolerance that people have towards the opinions of the interacting peers.
\subsection{Prior Results}
In \cite{gomez2012bounded,Hegselmann02opiniondynamics} it is proved that for any $\tau$ and $\lambda$, for time $t$ that goes to infinity, the Bounded Confidence Model converges almost surely to a stationary distribution that is either a Dirac measure at one single point (`Total Consensus') or a combination of Dirac measures at points separated by at least $\tau$ (`Partial Consensus'). 
\subsection{Our Contribution}
We show a similar result for the Attraction-Repulsion Model. For such model, the stable configurations (or the absorbing states) are those where the distance between any pair of points is either $0$ or $1$. We say that these are `Trivialized' configurations. They can be grouped into two categories: the configurations of `Polarization' and those of `Total Consensus' (Definition \ref{def:trivconf}). We show that for any $\tau, \lambda$ and $\mu$, as $t$ goes to infinity, the Attraction-Repulsion model converges to either a `Polarization' or a `Total Consensus' configuration (Figure \ref{fig:simtaumod}). Note that the configurations of `Partial Consensus' are not stable, unlike the Bounded Confidence Model. In the following, for brevity, we use the term `Consensus' to indicate `Total Consensus'.
\begin{defin}[Trivialized configuration] \label{def:trivconf}
We say that $\hat{\Phi} =(\hat{\phi}^1,...,\hat{\phi}^n) \in I^n$ is a trivialized configuration if one of the following conditions is satisfied.
\begin{itemize}
    \item Polarization: $\hat{\phi}^i \in \{0,1\} $ for every $i \in [n]$,
    \item Consensus: for some $\alpha \in I$, $\hat{\phi}^i = \alpha $ for every $i \in [n]$.
\end{itemize}
We denote by $\mathcal{T}_{P,n}$ the set of polarized configurations in $I^n$, by $\mathcal{T}_{C,n}$ the set of configurations of consensus in $I^n$ and by $\mathcal{T}_{n} = \mathcal{T}_{P,n} \cup \mathcal{T}_{C,n}$ the set of all trivialized configurations in $I^n$.
\end{defin}

\begin{defin}[Trivialization] \label{def:trivialization}
Let  $\left(\{\Phi_t\}_{t \in \mathbb{N}}, f, D_0 \right) $ be a random interaction process.
We say that the process trivializes if for any $\varepsilon > 0$, there exist $Y \in \mathcal{T}_n$ and $t_0 \in \mathbb{N}$ such that for any $t \geq t_0$
\begin{align}
\Vert \Phi_t -Y \Vert_{\infty} < \varepsilon,
\end{align}
where $ \Vert \Phi \Vert_{\infty} = \max_i |\phi_i|$ denotes the infinity norm.
\end{defin}
The probabilistic convergence will come later in Theorem \ref{thm:condtriv}.

\section{Trivialization for Finite Population}
Consider $f$ satisfying the following condition:
\begin{condition} \label{cond:fixpoints}
$(x_1,x_2)  $ is a fixed point of $f$ if and only if $|x_1-x_2| \in \{ 0,1\}$,
\end{condition}
which means that the set of stable configurations for the interaction process associated to $f$ corresponds to $\mathcal{T}_n$. One can argue that if Condition \ref{cond:fixpoints} holds, then the process trivializes almost surely. Unfortunately, this is not sufficient. Consider the following counterexample. Let $\tau_1,\tau_2$ be in $(0,1)$ and assume $\tau_1 <\tau_2$. Consider $f$ that satisfies Condition \ref{cond:fixpoints} and let $(x',y') = f(x,y)$ be such that if $|x-y|<\tau_1$, then $|x' -y'| <|x-y|$ (attraction); if $|x-y|>\tau_2 $, then $|x' -y'| >|x-y|$ (repulsion); if $\tau_1 \leq |x-y| \leq \tau_2 $, then $ |x'-y'| = |x-y|$ (for instance they shift by the same quantity towards the furthest border). The corresponding process does not trivialise with probability 1. 

This means that we need to add other conditions to Condition \ref{thm:condtriv} to guarantee trivialization of the process. In Theorem \ref{thm:condtriv} we show almost sure trivialization for any interaction function that has an attraction-repulsion cleavage property similar to Model \ref{ex:taumod}. The class of functions that satisfy the conditions of Theorem \ref{thm:condtriv} includes Model \ref{ex:taumod} and it is slightly more general, since it allows the interaction to be asymmetrical and non-linear. 

We provide a proof of the following result in Section \ref{sec:prooffinite}.
\begin{thm}
\label{thm:condtriv}
Let $f$ be an interaction function, let $D_0$ be a non-degenerate distribution in $I$ and assume $n <\infty$. Let $\left(\{ \Phi_t \}_{t \in \mathbb{N}}, f, D_0  \right) $ be the associated random process. Denote $(x',y') = f(x,y)$ and assume that $f$ satisfies the following attraction-repulsion condition: there exist $\tau \in I$, $C_f^A<1$ and $C_f^R >1$ such that for any $x,y \in I$ (assume $x<y$) 
\begin{itemize}
\item if $ |x-y| \leq \tau$, then $ x',y' \in  [x,y]$ and $  |x'-y'|\leq C_f^A |x-y|$ (attraction);
\item if $|x-y| > \tau$, then $ x',y' \in  I \setminus [x,y]$ and $|x'-y'| \geq C_f^R |x-y|$ (repulsion).
\end{itemize}

Then, the process trivializes with probability $1$.
\end{thm}
Note that the assumptions of Theorem \ref{thm:condtriv} include Condition \ref{cond:fixpoints}. However, we believe that the conditions of Theorem \ref{thm:condtriv} may not be necessary, and almost sure trivialization may be proved for a larger class of interaction functions.

\section{Proof of Theorem \ref{thm:condtriv}}
\label{sec:prooffinite}

For $X \in I^n$, let $\{\Phi_t^X\}_{t \in \mathbb{N}}$ be the random process associated to $f$ starting from configuration $X$ (i.e. $\Phi_0^X= X$).

Let $\varepsilon < \min \{\tau/2, (1-\tau)/2 \}$. Consider the set of states $\varepsilon$-close to a trivialized configuration:
\begin{align}
    A_{\varepsilon} := \{ X \in I^n : \min_{Y \in \mathcal{T}_{n}} \Vert X -Y \Vert_{\infty} < \varepsilon \}
\end{align}
Observe that $A_{\varepsilon}$ is an absorbing set, i.e. $\mathbb{P}(\Phi_1^X \in A_\varepsilon) =1$.


Define the set of ``promising'' states at $t$ steps:
\begin{align}
    V_{\varepsilon}(t) := \left\lbrace X \in A_{\varepsilon}^C : P^t(X,A_{\varepsilon}) \geq {n \choose 2}^{-t} \right\rbrace, \label{for:Vm}
\end{align}
where $P^t(X,A_{\varepsilon}) = \mathbb{P} (\Phi_t^X \in A_{\varepsilon})$ is the probability of going from $X \in I^n$ to somewhere in $A_{\varepsilon}$ in $t$ steps.

Notice that since $A_{\varepsilon} $ is absorbing, $V_{\varepsilon} (t) \subseteq V_{\varepsilon} (t+1) $, for any $t > 0$.
\begin{prop} \label{prop:uniftransience}
For any $t$, $V_\varepsilon(t)$ is uniformly transient.
\end{prop}
$V_\varepsilon(t)$ is uniformly transient if there exists $M_t < \infty$ such that  $\sum_{m=1}^{\infty} P^m(X,V_\varepsilon(t)) \leq M_t $ for all $X \in I^n$ (see e.g. \cite{meyn1993markov}).
Proposition \ref{prop:uniftransience} follows from results on Markov chains theory, a proof can be found in \cite{meyn1993markov,mscthesis}.

\begin{claim}
\label{claim:M}
There exists $T < \infty$, such that $V_\varepsilon(T) = A_{\varepsilon}^C$.
\end{claim}
Assume that the claim is true. Then, by Proposition \ref{prop:uniftransience}, $A_{\varepsilon}^C$ is uniformly transient, i.e. $\sum_{m=1}^{\infty} P^m(X,A_{\varepsilon}^C) <\infty$ for every $ X \in I^n$. By Borel-Cantelli Lemma, for every $X \in I^n$, $\mathbb{P}( \lim_{t \rightarrow \infty} \Phi_t^X \in A_{\varepsilon}^C) =0$ and consequently $\mathbb{P}( \lim_{t \rightarrow \infty} \Phi_t^X \in A_{\varepsilon}) =1$. This holds for any $0 <\varepsilon < \min \{\tau/2, (1-\tau)/2 \}$, thus the process trivializes almost surely.

It only remains to prove the claim.

\begin{proof}[Proof of Claim \ref{claim:M}]
For any finite sequence of pairs $ p_t = \{ (i,j)_t \}_{t \leq T} \in {[n] \choose 2}^T $ for $T \in \mathbb{N}$, let $F_t^{X,p}$ be the deterministic process that at each step $t$ is forced to choose $p_t$ as interacting pair of agents and such that $F_0^{X,p} = X$. It is enough to show that for any $X \in A_{\varepsilon}^C$, there exists $p_t$ such that $F_T^{X,p} \in A_{\varepsilon}$. 

In fact, $P^T(X,A_\varepsilon) \geq \mathbb{P}( \Phi_t^X = F_t^{X,p} \quad \forall t\leq T ) = {n \choose 2}^{-T}$, that implies that $X \in V_\varepsilon(T)$.

Let us define $A_{P,\varepsilon} := \{ X \in I^n : \min_{Y \in \mathcal{T}_{P,n}} \Vert X -Y \Vert_{\infty} < \varepsilon \} $ and $A_{C,\varepsilon} := \{ X \in I^n : \min_{Y \in \mathcal{T}_{C,n}} \Vert X -Y \Vert_{\infty} < \varepsilon \} $, the sets of states $\varepsilon$-close to a polarized and a consensus configuration respectively. Clearly, $A_{\varepsilon} = A_{P,\varepsilon} \cup A_{C,\varepsilon} $.

 Assume initially that $\tau < \frac{1}{2}$.
 Let $X \in A_{C,\tau/2}$, i.e. the distance between any pairs of agents is below $\tau$. Let $ p_t =  \argmax_{(k,l) \in {[n] \choose 2}} \{ |(F_t^{X,p})_k - (F_t^{X,p})_l | \}$, where $(F_t^{X,p})_k$ denotes the $k^{th} $ component of $F_t^{X,p}$. In words, at each step we choose the pair of agents separated by the maximum distance. After $n-1$ steps the maximum distance between any pairs of points decreases by a non vanishing quantity (at least $C_f^A$). Since the number of agents is finite, there exists $L < \infty$ such that $F_L^{X,p} \in A_{C,\varepsilon}$. 

 Let $X \not\in A_{C,\tau/2}$. There exists a pair $r,s \in [n]$ such that $|X_r -X_s| >\tau$, assume without loss of generality that $X_r < X_s$. Let $\delta$ be such that $0<\delta < \frac{1}{2} -\tau$ and $ \delta \leq \varepsilon $. Such $\delta$ exists since $\tau < \frac{1}{2}$. If we repeatedly pair agents $r$ and $s$, they will repulse each other and at each step their distance will increase at least by $C_f^R$. Hence, there exists $L_0 < \infty$, such that after $L_0$ steps, agents $r$ and $s$ are $\delta$-close to $0$ and $1$ respectively, and the other agents stay unchanged. Subsequently, for each $k\in [n] \setminus \{r,s \}$, we repeatedly choose pair $(k,s)$ if $X_k \leq 1/2$, and pair $(k,r)$ otherwise, until $k$ is $\epsilon$-close to $0$ or $1$. Notice that the pairs selected will always repulse, such that $r$ and $s$ stay close to the borders of the interval. Since $n$ is finite, after a finite number of steps such process will be in $A_{P,\varepsilon}$.


 Let $ \tau \geq \frac{1}{2}$. Let $G_\tau = \{ Z \in I^n : \exists i,j \in [n] \text{ such that } Z_i \sim Z_j \text{ and } |Z_i -Z_j| > \tau \}$, where $Z_i \sim Z_j$ denotes that there are no points between $Z_i $ and $Z_j$. In words, $G_\tau$ includes all the configurations that contain at least one gap larger than $\tau$ between two consecutive points. Assume $X \in G_{\tau}$. Then, at each step we can pair the two agents separated by the largest gap.
 After $n-1$ steps, the largest gap increases by a non vanishing quantity (at least $C_f^R$), hence in a finite number of steps there exists a gap larger than $1-\varepsilon$, thus the system is in $A_{P,\varepsilon}$.

 Let $X \not\in G_{\tau}$. At each step we choose the pair $\{(i,j)\}_{t}$ such that $ |(F_t^{X,p})_i - (F_t^{X,p})_j | = \max_{k,l \in [n]} \{ |(F_t^{X,p})_k - (F_t^{X,p})_l | :  |(F_t^{X,p})_k - (F_t^{X,p})_l |< \tau\}$, i.e.
 we choose the pair of agents whose distance is maximum, but constrained to be smaller than $\tau$. Then, after a finite number of iterations, the distance between any pair of points is either very small, or greater than $\tau$. Thus, either $F_L \in A_{C,\tau/2}$ or $F_L \in G_{\tau} $. The result then follows from previous arguments.
\end{proof}

\section{Trivialization for Infinite Population}

\subsection{One-dimensional model}
Consider the Attraction-Repulsion Model \ref{ex:taumod}. What happens when the population size $n$ tends to infinity?

Note that in the dynamics described in Section \ref{sec:ourmodel}, each individual is updated at rate $\frac{2}{n}$, which decreases with $n$. In this Section, we consider a slightly different dynamic. At each time step, we select a random matching of the agents (instead of a random pair), and we move each pair of agents independently according to the interaction function described in Model \ref{ex:taumod}. We assume the number of agents to be even. In this way, each individual is updated at rate 1, independently of $n$.

Let $p_P$ be the probability that the process polarizes. Formally, we can define $p_P := \mathbb{P}\{ \exists t: \Phi_t \in A_{P,b} \}$, with $b = \min\{ \tau/2, (1-\tau)/2 \}$.
We carried out some experiments to simulate $p_P$ depending on the threshold $\tau$ and the number of agents $n$. Figure \ref{fig:probpolarization} shows a plot obtained by running a Monte-Carlo simulation for the Attraction-Repulsion model with $\lambda = \mu = 0.5$. We notice that the probability of polarization decreases as $\tau$ increases. We observe that as $n$ increases, $p_P(\tau)$ tends to a step function with transition phase at $\tau \approx 0.52$.

\begin{figure}[h]
\centering
\includegraphics[width = 0.30\textwidth]{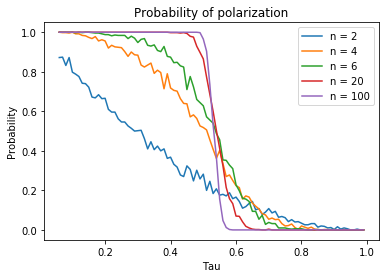}
\caption{Probability of polarization depending on the threshold $\tau$, and the number of agents $n=2,4,6,20,100$, for the Attraction-Repulsion model with $\lambda = \mu = 0.5$. 
}
\label{fig:probpolarization}
\end{figure}
Let $f_t^{(n)}$ be the empirical distribution of the $n$ points at time $t$. As $n$ gets larger, $f_t^{(n)}$ tends to a continuous distribution $f_t$, that satisfies a deterministic PDE that accounts for the variations in the distribution of the agents after each interaction.
Based on numerical approximations, we conjecture that for any $\lambda, \mu $ there exists $\tau_C$ such that 
\begin{align*}
     \lim_{t \rightarrow \infty} \lim_{n \rightarrow \infty} f_t^{(n)}(x) = \left\{
                \begin{array}{ll}
                 \frac{1}{2} \delta(x) + \frac{1}{2} \delta (x-1)&\text{  if  }  \tau < \tau_C,\vspace{0.2cm}\\
                 \delta(x- \frac{1}{2}) &\text{  if  }  \tau > \tau_C,
                \end{array}
              \right.
\end{align*}
where both limits are in distribution and where $\delta (x)$ denotes the Dirac Delta measure centered at $0$.

If $\Phi_0 \distas{iid} f_0 := \mathcal{U}([0,1])$, then $\lim_{n \rightarrow \infty} f_0^{(n)} = f_0$ in distribution. Moreover, the continuous densities $f_t(x) $ satisfy the equation
\begin{align}
\label{for:PDE}
\frac{\partial f_t(x)}{\partial t} &= \frac{1}{1-\nu} \int_{x-(1-\nu) \tau}^{x+ (1-\nu) \tau} f_t \left(\frac{x-\nu y}{1-\nu}\right) f_t(y) dy \nonumber \\
& +\frac{1}{1-\mu} f_t\left( \frac{x-\mu}{1-\mu} \right) \int_{0}^{\frac{x-\mu}{1-\mu}-\tau}f_t(y) dy \\
& + \frac{1}{1-\mu} f_t \left(\frac{x}{1-\mu}\right) \int_{\frac{x}{1-\mu} +\tau}^{1} f_t(y) dy - f_t(x), \nonumber
\end{align}
where we wrote $\nu = \frac{\lambda}{2}$ for brevity.
For the positive terms, if an agent is in state $\frac{x-\nu y}{1-\nu}$ and interacts with another agent in $y$, for $y \in [x-(1-\nu) \tau, x+(1-\nu) \tau]$ (i.e. within distance $\tau$), it moves to state $x$. Moreover, an agent in $  \frac{x-\mu}{1-\mu}$ can interact with any agents in the interval $[0,   \frac{x-\mu}{1-\mu}-\tau] $ and be repelled to $x$, or an agent in $ \frac{x}{1-\mu}$ can be matched with an individual in $[\frac{x}{1-\mu} +\tau,1] $ and move to $x$. The negative term follows since a point in $x$ will move away from $x$ after an interaction with any other point in $[0,1] \setminus \{x \}$. Hence, at each $t$ we assume the $n$ agents to be iid realizations from $f_t$.

As far as we know, there is no explicit solution to equation (\ref{for:PDE}). We approximated the solution numerically, using a forward Euler method. More specifically, recursively for any time step $t$, we computed $f_t(x) $ in a discretized subset of $I$ as
\begin{align}
f_{t+1}(x) = f_t(x) + \frac{\partial f_t(x)}{\partial t}.
\end{align}
We then created a piecewise constant approximation of  $f_t(x) $ in the rest of the interval, to use in the subsequent iteration. In \cite{gomez2012bounded} a similar numerical approach is described for the Bounded Confidence Model.
It is easy to show that if $f_0$ is a continuous probability density, then $f_t$ is a continuous probability density for any $t$. Moreover, if $f_0$ is symmetric with respect to $\frac{1}{2}$, then $f_t$ preserves this symmetry at every step.

We noticed that as $t$ goes to infinity, there exists a $\tau_C$ such that $f_t$ converges to a symmetric polarized configuration if $\tau < \tau_C$, and $f_t$ converges to consensus at $\frac{1}{2}$ if $\tau > \tau_C$. The value of $\tau_C$ depends on $\lambda$ and $\mu$. For instance, if $\lambda = \mu = 0.5$, $\tau_c \approx 0.526$ (Figure \ref{fig:hist2}).
\begin{figure}[h]
\centering
\includegraphics[width = 0.16\textwidth]{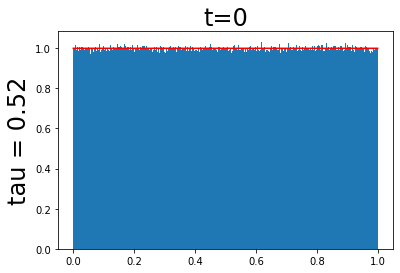}
\includegraphics[width = 0.15\textwidth]{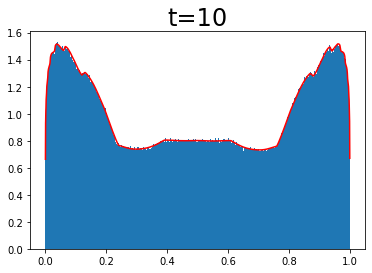}
\includegraphics[width = 0.15\textwidth]{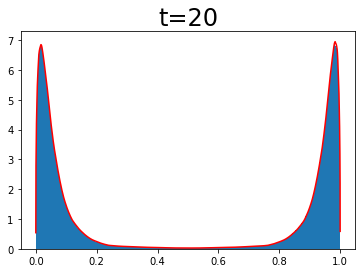}
\includegraphics[width = 0.159\textwidth]{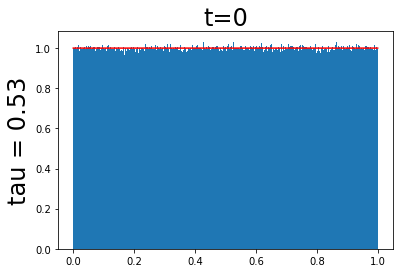}
\includegraphics[width = 0.15\textwidth]{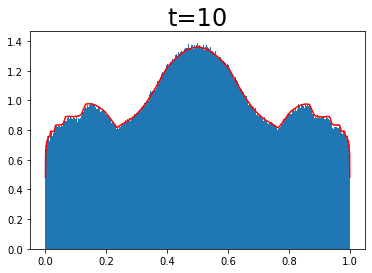}
\includegraphics[width = 0.15\textwidth]{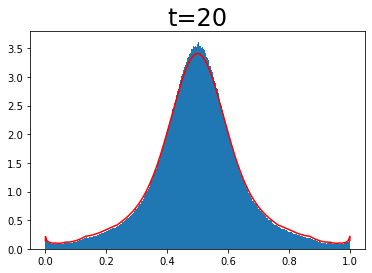}
\caption{$\lambda = \mu = 0.5$. Evolution of $f_t$ at $t=0,10,20$ for $\tau = 0.52$ (above) and $\tau = 0.53$ (below). The red line is the approximation of $f_t$ obtained through a forward Euler method, the blue shadow is an histogram obtained running a simulation of the process with $n = 10^{7}$ agents.}
\label{fig:hist2}
\end{figure}
\subsection{D-dimensional models}
When considering opinion dynamics, the one-dimensional model represents a singular opinion that polarizes or agrees under pairwise interaction.
However, we would like to consider a model that accounts for multiple different topics discussed jointly within one interaction.

We create a D-dimensional hypercube to represent D topics for opinion interaction. Agents are encoded as D dimensional vectors in the unit hypercube, i.e. for $i \in [n]$ and $ t \in \mathbb{N}$, $\bm{\phi}_t^i \in [0,1]^D$ denotes the opinion of agent $i$ at time $t$, and $\bm{\Phi_t} = (\bm{\phi_t^1},...,\bm{\phi_t^n}) \in [0,1]^{D \times n}$ denotes the state of the system at time $t$. The movement function takes in two points and maps them along the line that passes through the pair of points. If the Euclidean distance between the points is less than $\tau$, the points move towards each other along the line and their distance is decreased by a factor $\lambda$. If the distance between the two points is larger than $\tau$, the points move away from each other, towards the borders of the hypercube along the line, and the distance between each point and the respective intersection of the line and the boundary is decreased by a factor $\mu$. 

We simulate this interaction on a two dimensional square and find polarization and consensus convergence behavior analogous to the one-dimensional unit interval. We show two simulation results of this interaction function in Figure~\ref{fig:ani}.
\begin{figure}[h]
   
    \centering
    \includegraphics[width = 0.12\textwidth]{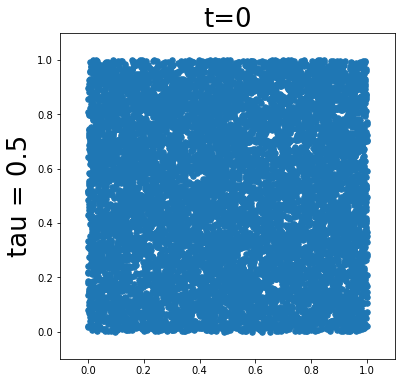}
    \includegraphics[width = 0.11\textwidth]{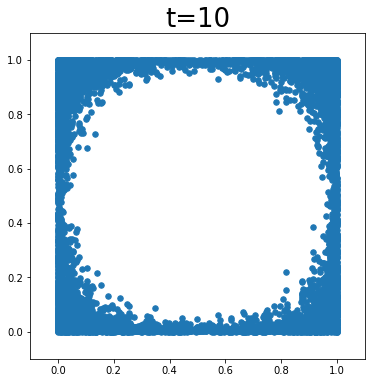}
    \includegraphics[width = 0.11\textwidth]{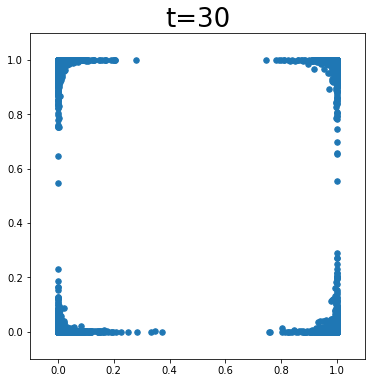}
    \includegraphics[width = 0.11\textwidth]{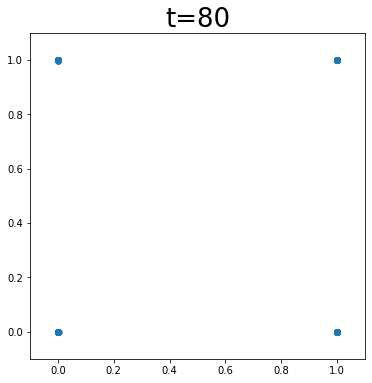}

    \includegraphics[width = 0.12\textwidth]{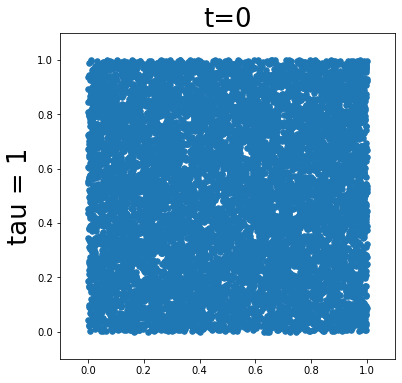}
    \includegraphics[width = 0.11\textwidth]{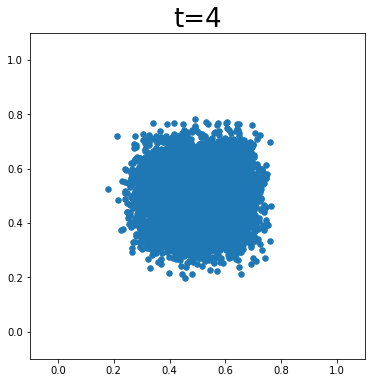}
    \includegraphics[width = 0.11\textwidth]{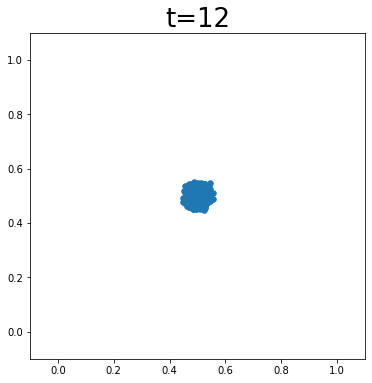}
    \includegraphics[width = 0.11\textwidth]{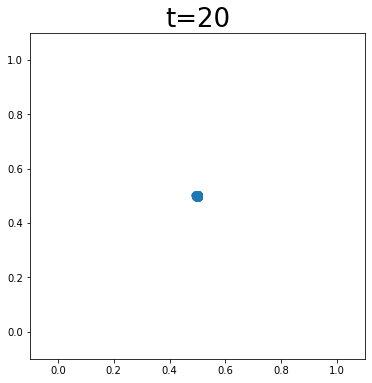}
    \caption{Simulations of points on a square space over a $\tau$ = 0.5 threshold (above) and over a $\tau$ = 1 threshold (below), with $\lambda =\mu = 0.5$.}
     \label{fig:ani}
\end{figure}
With this model on the unit square, consensus leads to one cluster of mild opinions inside the square, while polarization leads to clusters of combinations of extreme opinions along the borders. For $n$ large enough, polarization happens on the four corners of the square, with approximately $n/4$ agents in each corner, and consensus happens in the middle of the square. We conjecture that for $t$ that goes to infinity and for $n$ that tends to infinity,  $f_t^{(n)}(\vec{x})$, the empirical distribution of the $n$ points at time $t$, converges to either a polarization or a consensus state, with transition phase phenomenon similar to the one-dimensional case.

We infer similar behavior for higher dimensions, with clustering in the borders of the hypercube under polarization, or one cluster inside the hypercube under consensus. This can model radical party formation, and how social interaction and groupthink can enhance extreme opinions.

\section{Future Directions}
\subsection{A martingale approach}
We propose an alternative approach for proving Theorem \ref{thm:condtriv}.
\begin{thm}
\label{thm:martingale}
Let $f$ be an interaction function and let $\left(\{ \Phi_t \}_{t \in \mathbb{N}}, f, D_0  \right) $ be the random process associated to it, with initial distribution $D_0$. Assume that there exists a function $h:[0,1]^n \rightarrow \mathbb{R}$ such that $\{h(\Phi_t)\}_{t \in \mathbb{N}}$ is a non-negative super-martingale (or a bounded sub-martingale) with respect to the canonical filtration and that for any $\varepsilon >0$
\begin{align}
|h(\Phi_{t+1})-h(\Phi_t)| < \varepsilon \implies \min_{Y \in \mathcal{T}_{n}} \Vert \Phi_t -Y \Vert_{\infty} < \delta(\varepsilon) , \label{for:condmg}
\end{align}
where $\delta(\varepsilon)$ is such that $\lim_{\varepsilon \rightarrow 0^+} \delta(\varepsilon) = 0$. Then the process trivializes almost surely.
\end{thm}
\begin{proof}
By Doob's Martingale Convergence Theorem \cite{kushner1971introduction}, $h(\Phi_t)$ converges almost surely to a random variable, as $t \rightarrow \infty$. Hence, it is Cauchy almost surely, i.e. for any $\varepsilon >0$ there exists $t_0$ such that for any $t\geq t_0$, $|h(\Phi_{t+1}) -h(\Phi_t)| < \varepsilon$. By (\ref{for:condmg}), $\min_{Y \in \mathcal{T}_{n}} \Vert \Phi_t -Y \Vert_{\infty} < \delta(\varepsilon)$ for any $t\geq t_0$, thus the process trivializes almost surely.
\end{proof}
\begin{conjecture}
For any $f$ that satisfies the conditions of Theorem \ref{thm:condtriv} and for any finite $n$, there exists a function $h:[0,1]^n \rightarrow \mathbb{R}$ such that the conditions of Theorem \ref{thm:martingale} hold.
\end{conjecture}
For example, consider $h(\phi) = \sum_{(i,j) \in {[n] \choose 2}} |\tau - |\phi_i-\phi_j||$. For $n=3$, one can show that $\{h(\Phi_t) \}_{t \in \mathbb{N}}$ is a bounded sub-martingale. However, this does not hold for $n\geq 4$.

\subsection{Other problems}
In Theorem \ref{thm:condtriv} we defined assumptions on $f$ that guarantee almost sure trivialization. The next step is to extend these assumptions to a larger class of interaction functions.

However, it is interesting to study problems related to the Attraction-Repulsion model specifically. For instance, computing the probability of polarization against the probability of consensus depending on the threshold $\tau$ and on the number of agents $n$ (Figure \ref{fig:probpolarization}) could be useful for several applications. For $n=2$, the computation is straightforward, since the long-run behavior depends uniquely on the initial state; for larger $n$ it becomes trickier. Moreover, it is interesting to compute the expected mixing time, depending on the parameters of the model. We expect it to depend on $\tau, \lambda, \mu $ and $n$.

Another direction is extending the D-dimensional model to a larger class of convex domains, that are not necessarily hypercubes, in order to capture a wider range of dynamics of opinions. We expect the process to converge either to consensus in one cluster of mild opinions, or to polarization towards the boundary. We ran some experiments on the unit circle. For a large number of agents, we observed that points either merge to the center of the circle, or are pushed towards the border and move along the circumference. When points are close to the circumference, the process becomes similar to the Bounded Confidence Model, since the opinions can be attracted to each other if at distance less than $\tau$, but the repulsion effect is less significant. We noticed that they eventually separate into clusters along the border of distance at least $\tau$ (Figure \ref{fig:circ}).

\begin{figure}[h]
    \centering
    \includegraphics[width = 0.11\textwidth]{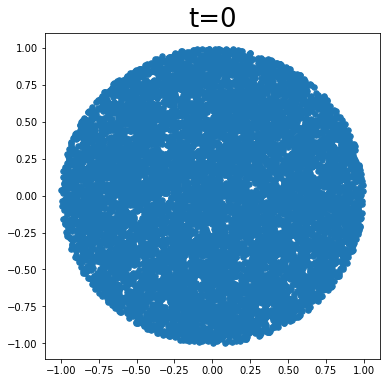}
    \includegraphics[width = 0.11\textwidth]{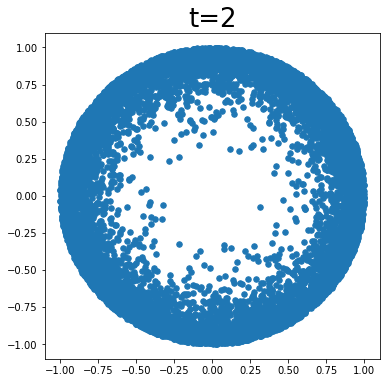}
    \includegraphics[width = 0.11\textwidth]{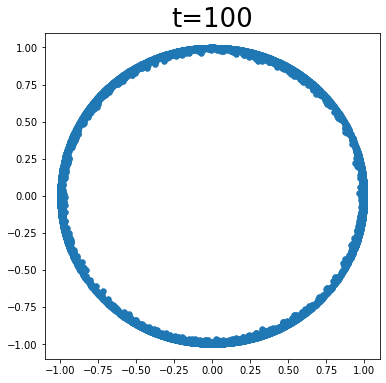}
    \includegraphics[width = 0.11\textwidth]{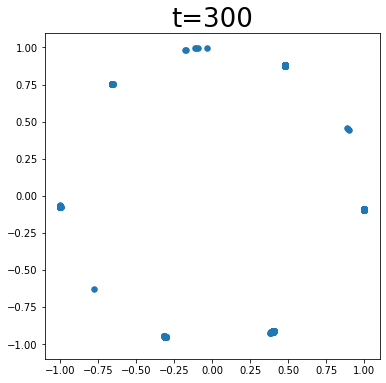}

    \caption{Simulation of points on a unit circle space over a $\tau$ = .5 threshold.}
    \label{fig:circ}
\end{figure}




\enlargethispage{-1.2cm}
\IEEEtriggeratref{9}


\bibliographystyle{IEEEtran}
\bibliography{references.bib}

\end{document}